\documentclass[10pt]{article}
\usepackage{xcolor}
\usepackage{amsmath,amssymb,amsthm}
\usepackage{nicefrac}
\def\x{{\mathbf x}}

\def\X{{\mathbf X}}
\def\x{{\mathbf x}}

\def\a{\mathbf{a}}

\newtheorem{definition}{Definition}
\newtheorem{proposition}{Proposition}
\newtheorem{corollary}{Corollary}
\newtheorem{theorem}{Theorem}
\newtheorem{example}{Example}
\newtheorem{remark}{Remark}
\usepackage{pgf}
\usepackage{cite}
\usepackage{pgfplots}
\usepackage{tikz}
\usetikzlibrary{arrows}
\usetikzlibrary{chains,positioning,arrows,calc}
\usetikzlibrary{arrows,shapes,automata,backgrounds,petri,patterns}
\usetikzlibrary{external}
\usepackage[onehalfspacing]{setspace}
\newcommand{\reff}[1]{(\ref{#1})}

\newcommand\as{\textsf{as}}
\renewcommand\Pr[1]{\text{Pr}\left[{#1}\right]}
 \newcommand\sgn[1]{\textsf{sgn}\left({#1}\right)}

\usepackage[labelfont=bf,labelsep=period,justification=raggedright]{caption}

\newcommand\vect[1]{\ensuremath{\mathbf{#1}}}

\begin{document}
\begin{flushleft}
{\Large
\textbf{Bounds on the Average Sensitivity of Nested Canalizing Functions}
}
% Insert Author names, affiliations and corresponding author email.
\\
Johannes Georg Klotz$^{1}$, 
Reinhard Heckel$^{2}$,
Steffen Schober$^{1}$
\\
\bf{1} Institute of Communications Engineering, Ulm University, Ulm, Germany\\
\bf{2} Department of Information Technology and Electrical Engineering, ETH Zurich, Switzerland
\\
$\ast$ E-mail: johannes.klotz@uni-ulm.de
\end{flushleft}

\section*{Abstract}
Nested canalizing Boolean (NCF) functions play an important role in biological motivated regulative networks and in signal processing, in particular describing stack filters. It has been conjectured that NCFs have a stabilizing effect on the network dynamics. It is well known that the average sensitivity plays a central role for the stability of (random) Boolean networks. Here we provide a tight upper bound on the average sensitivity for NCFs as a function of the number of relevant input variables. As conjectured in literature this bound is smaller than $\frac{4}{3}$. This shows that a large number of functions appearing in biological networks belong to a class that has very low average sensitivity, which is even close to a tight lower bound.

\section*{Introduction}
Boolean networks play an important role in modeling and understanding signal transduction and regulatory networks.
Boolean networks have been widely studied under different point of views, e.g. \cite{A03,SK04,B08}. 
One line of research focuses on the dynamical stability of randomly created networks.
For example,  random Boolean networks tend to be  unstable if 
the functions are chosen from the set of all Boolean functions with average number of variables (average \emph{in-degree}) larger than two \cite{K69}.
This can be attributed to the fact that the expected \emph{average sensitivity} of random functions, with in-degree $> 2$ larger than one. The expected average sensitivity is an appropriate measure for the stability of random Boolean networks \cite{L07,SB07}.

If only functions from certain classes are chosen, stable behavior can be achieved for higher in-degrees.
For instance, canalizing and nested canalizing functions,   
introduced in \cite{K93, KPS03}, have been conjectured \cite{KPST04} to have a stabilizing effect on network dynamics. In \cite{P10} it has been shown that Boolean networks can be stable even if the average in-degree is high.
Interestingly, studies of regulatory network models have shown that a large number of their functions are canalizing \cite{W42, HSW02,LLL04,DB08,KFS13,ML11}.
Canalizing functions are also important for the construction of stack filters used in signal processing \cite{GYC92}.

A Boolean function $f(x_1,\ldots,x_n)$ is \emph{canalizing} in variable $x_i$, if $f$ is constant when $x_i$ is set to its \textit{canalizing} value.
%In \cite{KFS} these functions have been investigated using Fourier analysis, and some spectral properties have been shown. 
%Further investigations have been undertaken, e.g., the calculation of the actual size of the class \cite{JSK04}.
Nested canalizing functions are canalizing functions, whose restriction to the non-canalizing value is again a canalizing function and so on 
(a precise definition is given later).
In this paper we analyze nested canalizing functions (NCFs) in particular their average sensitivities.
The notion of sensitivity was first introduced by Cook et al. \cite{CDR86}. It was applied later to Boolean functions \cite{BKS99} 
and can be viewed as a metric for the influence of a random permutation of the input variables on the output of the function. The average sensitivity was investigated  in \cite{S05} in the context of monotone Boolean functions. An upper bound for locally monotone functions was presented in \cite{KFS13}.
Here we give a tight upper bound on the average sensitivity of NCFs.
Our result shows that the average sensitivity of NCFs is always smaller than $\frac{4}{3}$ as conjectured in \cite{LAM13}.
We further provide a recursive expression of the average sensitivity and the zero Fourier coefficient of a NCF.
Finally we will discuss and compare our new bounds to bound in literature.

Our main tool is the Fourier analysis \cite{B61,FJS91} of Boolean functions, which is introduced in Section \emph{Notation, Basic Definitions and Fourier Analysis of Boolean Functions}, where we also address further concepts needed.
In Section \emph{Nested Canalizing Functions} spectral properties  of canalizing and NCFs are broached. 
Additionally we discuss functions, in which all variables are most dominant,
as they turn out to minimize the average sensitivity. In Section \emph{Average Sensitivity} the new bounds on the average sensitivity are presented based 
on a recursive expression of the average sensitivity of NCFs. 
We conclude with a discussion of the results and some final remarks.

\section*{Methods}
\subsection*{Notation, Basic Definitions and Fourier Analysis of Boolean Functions} 
A Boolean function (BF) $f \in \mathcal F_n = \{ f : \Omega^n  \rightarrow \Omega\}$ with $\Omega=\{-1,+1\}$ maps n-ary input tuples to a binary output. Note, that we choose the $\{-1,+1\}$-representation of the Boolean states instead of the $\{0,1\}$-representation, since it will turn out to be advantageous as it simplifies our calculations, especially in the Fourier domain. However, our results apply for all binary alphabets $\Omega$.

In general not all input variables have an impact on the output, i.e., are relevant.
\begin{definition} \cite{LAM13}
A variable $i$ is relevant to a BF $f$, if there exists an $\x \in \Omega^n$ such that
$$
f(\x) \neq f(\x \oplus e_i),
$$
where $x\oplus e_i$ is the vector obtained from $\x$ by flipping its $i$-th entry.\\
Further we define $rel(f) \subseteq [n] = \{1,2,\ldots,n\}$ as the set containing all relevant variables of $f$.
\end{definition}

\subsubsection*{Fourier Analysis of Boolean Functions}
In this section we recall basic concepts of Fourier analysis of BFs and some results from \cite{KFS13} concerning restrictions of BFs. Consider $\x=(x_1, x_2, \ldots, x_n )$ as an instance of a uniformly distributed random vector 
$\X=(X_1, X_2, \ldots, X_n ) \in \Omega^n$, i.e., the probability density function of $\X$ can be written as 
\begin{equation*}
\Pr{\X=\x} = \frac{1}{2^n}.
\end{equation*}
For $U \subseteq [n]$ we define the functions $\chi_U(\x)$ by
\begin{equation} \label{eq:chi}
\chi_U(\x)= 
\begin{cases}
\prod_{i \in U} x_i &\text{, if  $U \neq \emptyset$}\\
1 &\text{, if  $U = \emptyset$}
\end{cases}.
\end{equation}
Note that for $A \subset U$ and $\bar{A} = U \setminus A$ 
\begin{equation*}
\chi_U(\x)= \chi_A(\x)\cdot \chi_{\bar{A}}(\x),
\end{equation*}
which directly follows from the definition of  $\chi_U$ (Eq.~\reff{eq:chi}).

It is well known that any BF $f$ can be represented by its Fourier-expansion \cite{B61,FJS91},
\begin{equation} \label{eq:FC}
f(\x) = \sum_{U \subseteq [n]} \hat{f}(U)\cdot \chi_U(\x),
\end{equation}
where $\hat{f}(U)$ are the Fourier coefficients, given by
\begin{equation} \label{eq:fourier}
\hat{f}(U) = 2^{-n}\sum_{\x } f(\x)\cdot \chi_U(\x).
\end{equation}
\begin{example}
The following tables contain the truth-table representation, and the polynomial representation, i.e. Eq.~\reff{eq:FC}, of three well-known BFs, namely AND, OR, and XOR.
\begin{table}[!h]
\centering
\begin{tabular}{c c|c|c|c}
$x_1$ &$x_2$ & AND & OR & XOR \\
\hline
$-1$ &$-1$ & $+1$& $+1$&$ +1$  \\
$-1$ &$+1$ & $+1$& $-1$&$ -1$ \\
$+1$ &$-1$ & $+1$& $-1$&$ -1$ \\
$+1$ &$+1$ & $-1$& $-1$&$ +1$ \\
\end{tabular}
\caption{Truthtable representation}
\end{table}

\begin{table}[!h]
\centering
\begin{tabular}{c|c|c}
AND & OR & XOR\\ \hline
$\nicefrac{1}{2}+\nicefrac{1}{2}x_1+\nicefrac{1}{2}x_2-\nicefrac{1}{2}x_1x_2$ &$-\nicefrac{1}{2}+\nicefrac{1}{2}x_1+\nicefrac{1}{2}x_2+\nicefrac{1}{2}x_1x_2$  &$x_1x_2$
\end{tabular}
\caption{Polynomial representation (Eq.~\reff{eq:FC})}
\end{table}
\end{example}
\begin{remark}
The polynomial representation in the previous example is different to the one used in \cite{LAM13}, where the variables $x_i$ are defined over GF(2), i.e. $\Omega=\{0,1\}$, where addition ($\oplus$) and multiplication are defined modulo $2$. In this case, the AND function becomes $x_1x_1$, the OR function is given by $x_1\oplus x_2\oplus x_1x_2$ and XOR by $x_1\oplus x_2$.
\end{remark}
\subsubsection*{Restrictions of Boolean Functions}
We call a function $f^{(i,a_i)} \in \mathcal F_n $ a restriction of f, if it is obtained by setting the $i$-th input variable of $f$ to some constant $a_i \in \{-1,+1\}$. Every BF can be decomposed in two unique restricted functions for each relevant variable, as stated by the following proposition:

\begin{proposition}\label{prop:res1} 
    For any $f \in \mathcal{F}_n$ and each $i\in [n]$ 
    there exist unique functions  $f^{(i,+)}, f^{(i,-)} \in \mathcal{F}_n$, with $ i \notin rel(f^{(i,+)})$ and $i \notin rel(f^{(i,-)})$,  such that 
    $$
    f= g^{(i,+)} f^{(i,+)} + g^{(i,-)}f^{(i,-)},
    $$
    where the functions $g^{(i,+)}, g^{(i,-)} \in \mathcal F_n$ are given by 
    $$g^{(i,+)} (\vect x) = 
    \begin{cases}
 1 & \text{if}  ~ x_i  = 1\\
 0 & \text{else}
    \end{cases}
    \quad\text{and}\quad
    g^{(i,-)} (\vect x) = \begin{cases}
 1 & \text{if}  ~ x_i  = -1\\
 0 & \text{else}
    \end{cases}.
    $$
\end{proposition}

The Fourier coefficients of $f^{(i,+)}$ and $f^{(i,-)}$ can be obtained as stated in the following proposition.

\begin{proposition}\cite{KFS13} \label{prop:restriced}
Let $f$ be a BF in $n$ uniformly distributed variables. 
The Fourier coefficients of ${f}^{(i,a_i)}$ are given by
$$
    \hat{f}^{(i,a_i)}(U) = \hat{f}(U) +  a_i \cdot  \hat{f}(U \cup \{i\})
$$
where $U \subseteq [n]\setminus \{i\}$. 
\end{proposition}
The reverse relation, i.e., the composition of a BF by two restricted functions, is described in terms of Fourier coefficients by the following proposition.

\begin{proposition} \cite{KFS13}\label{prop:res2} 
The Fourier coefficients of a BF $f$ with uniform distributed input variables can be composed of the Fourier coefficients of its two restricted functions $f^{(i,-)}$ and $f^{(i,+)}$ according to
    \begin{equation*}
    \hat f(U) =  \frac{1}{2}\left(\hat f^{(i,+)}( U \setminus \{i\}) + (-1)^{| U \cap \{i\}|  } \hat f^{(i,-)}( U \setminus \{i\})  \right),
    \end{equation*}
    or
    \begin{equation*}
    \hat 2 f(U) = \begin{cases} 
 \hat f^{(i,+)}(U \setminus \{i\} ) + \hat f^{(i,-)}(U \setminus \{i\} ) &\text{if}~i \in U \\
 \hat f^{(i,+)}(U) - \hat f^{(i,-)}(U ) &\text{if}~i \notin U \\
    \end{cases}.
 \label{eq:prop}
    \end{equation*}
\end{proposition}

A immediate corollary of Proposition \ref{prop:res2} shows that the zero coefficient of a function can be constructed by the zero coefficients of the restricted functions:
\begin{corollary} \label{cor:zero}
The zero Fourier coefficient of any Boolean function $f$ can be written as:
    \begin{equation}
    \hat f(\emptyset) =  \frac{1}{2}\hat f^{(i,+)}(\emptyset)   +   \frac{1}{2}\hat f^{(i,-)}( \emptyset),
    \end{equation}
    where $i \in [n]$ is the index of some variable.
\end{corollary}

If we restrict a function to more than one variable, namely to a set of variables $K$, we denote the restricted function with $f^{(K,\a)}$, where $\a$ is a vector containing the values to which the functions is restricted. The Fourier coefficients of $f^{(K,\a)}$ are given by the following proposition:
\begin{proposition}\cite{KFS13}\label{[prop:gen_restriced}
Let $f$ be a Boolean function and $\hat{f}(U)$ its Fourier coefficients. 
Furthermore, let $K$ be a set containing the indices $i$ of the input variables $x_i$, 
which are fixed to certain values $a_i$. The Fourier coefficients of the restricted function $f^{(K,\a)}$ are then given as:
\begin{equation}\nonumber
\hat{f}^{(K,\a)}(U) = \sum_{S \subseteq K} \left( \Phi_{S}(\a) \cdot  \hat{f}(U \cup S) \right),
\end{equation}
where $\a$ is the vector with entries $a_i, i \in K$.
\end{proposition}

\subsection*{Nested Canalizing Functions}
In order to define nested canalizing functions (NCF) we first need the following definition:
\begin{definition}\label{def:canalizing}
A BF $f$ is called  $<i:a:b>$ canalizing if there exists a canalizing variable $x_i$ and a constant $a \in \{-1,+1\}$
such that
\begin{equation*}
f(\x|_{x_i=a}) = f^{(i,a)}(\x)=  b,
\end{equation*}
for all $x_1, ... x_{i-1}, x_{i+1} .... x_n$, where $b \in \{-1,+1\}$ is a constant.
\end{definition}
Hence, $f$ is canalizing in variable $i$ if the decomposition according to Proposition \ref{prop:res1} results in either $f^{(i,+)}$ or $f^{(i,-)}$ being a constant function.

As shown in \cite{KFS13} the Fourier coefficients of a canalizing function satisfy
\begin{equation}\label{eq:cond:1}
\hat{f}(\emptyset) +  a_i  \cdot  \hat{f}(\{i\}) = b_i.
\end{equation}

A NCF can be described recursively as a canalizing function, whose restriction is again a NCF or more formally:

%\begin{definition}\label{def:nested_canalizing}
%A BF $f$ with $k$ relevant variables is called $\{\pi :\alpha :\beta\}$ nested canalizing for the variable order $x_{\pi(1)},\ldots, x_{\pi(k)}$, if $f$ is $<\pi(1):\alpha_1:\beta_1>$ canalizing and the restricted function, which is obtained by setting $x_{\pi(1)}=\bar{a}_i$, is again nested canalizing.
%If $f$ is nested canalizing, the decomposition according to Proposition \ref{prop:res1} results in either $f^{(i,+)}$ or $f^{(i,-)}$ being a constant function and the other function is again a nested canalizing function.
%\end{definition}
%
%Another definition.
\begin{definition}\label{def:nested_canalizing2}
For $k=1$ and $k=0$ any BF with $k\leq n$ relevant variables is a NCF.
For $k>1$ a BF is a NCF if there exists at least one variable $i$ and constants $\alpha_i, \beta_i \in \{+1,-1\}$ such that $f^{(i,\alpha_i)}=\beta_i$  and  $f^{(i,-\alpha_i)}$ is a NCF with $k-1$ relevant variables.
\end{definition}

Let $x_{\pi(1)},\ldots, x_{\pi(k)}$ be the variable order for which a NCFs fulfills the properties from Definition \ref{def:nested_canalizing2}, then we call, following \cite{LAM13}, such a function $\{\pi :\alpha :\beta\}$ nested canalizing.

As shown in \cite{KFS13} the spectral properties of NCFs are given as:\\
$f$ is $\{\pi :\alpha :\beta\}$ nested canalizing, if for all $j \in \{1,\ldots,k\}$
\begin{equation}\nonumber
\sum_{S \subseteq [j]} \left(\alpha_j^{|S \cap \{j\}|} \cdot \chi_{S \setminus \{j\}}(\bar \alpha) \cdot  \hat{f}(\tilde{S}) \right)=\beta_j,
\end{equation}
where $\bar \alpha$ is a vector containing all negated $\alpha_i$, i.e. $\bar \alpha_i = - \alpha_i$ and $
\tilde{S}$ is a set which is retrieved by applying the permutation $\pi$ to the elements of $S$.

For illustration, consider the following example:
\begin{example}
Let $f$ be $\{\pi: \alpha: \beta\}$ NCF with $k=2$ relevant variables and $\pi$ such that $\tilde{S} = \pi(S) = S$, then
\begin{align*}
\beta_1 &= \hat f(\emptyset) + \alpha_1 \hat f(\{1\}) \\
\beta_2 &= \hat f(\emptyset) - \alpha_1 \hat f(\{1\}) + \alpha_2 \hat f(\{2\}) - \alpha_1 \alpha_2 \hat f({1,2}).
\end{align*}
\end{example}
\subsubsection*{Properties of Nested Canalizing Functions}
In this section we state some properties of NCF.
First we address most dominant variables, which are defined as follows:
\begin{definition}
\cite[Def 4.5]{LAM13} Variable $i$ is called a most dominant variable of $f$ if there exists at least one variable order starting with this variable, i.e., there exists a permutation $\pi$ such that $\pi(1)=i$, for which $f$ is $\{\pi :\alpha :\beta\}$ nested canalizing.
\end{definition}
The set of most dominant variables has an impact on a number of Fourier coefficients, which is summarized in the following proposition.
\begin{proposition}\label{prop:most}
Let $K$ be the set of most dominant variables of a $\{\pi :\alpha :\beta\}$ NCF $f$. Then the absolute values of the corresponding Fourier coefficients are all equal , i.e., $\forall U \subseteq K, U \neq \{\emptyset\}$
\begin{equation*}
\left|\hat{f}(U) \right| = c, \hspace{0.8cm} c > 0,
\end{equation*}
or, more general,
\begin{equation}\label{eq:prop:most:2}
\alpha_j \cdot \chi_{U \setminus \{j\}}(\bar{\alpha}) \cdot \hat{f}(U) = c \hspace{1cm} \forall U\subseteq K, U \neq \{\emptyset\} \text{ and } \forall j \in K.
\end{equation}
Furthermore
\begin{equation*}
|\hat{f}(\emptyset)| =1-c,
\end{equation*}
and
\begin{equation*}
\beta_i = \sgn{\hat{f}(\emptyset)} \forall i \in K.
\end{equation*}

\end{proposition}
\begin{proof}

The proof for the zero and first order coefficients, i.e. $|U| = 1$ and $U=\emptyset$, follows directly from Eq. \reff{eq:cond:1}. We can hence use Eq. \reff{eq:prop:most:2} as an induction hypothesis for coefficients with order smaller than $|U|$. We show next that as a result this is also valid for coefficients with order $|U|+1$.
%Lets now consider coefficients $\tilde{U}=U \cup \{k\},k \in K, k \notin U$  with order assuming Eq. \reff{eq:prop:most:2} is correct for coefficients with :

Using
\begin{equation}\nonumber
\hat{f}^{(U,\bar \alpha)}(T) = \sum_{S \subseteq U} \left(\chi_{S}(\bar{\alpha} )\cdot  \hat{f}(T \cup S) \right),
\end{equation}
and that $f$ is canalizing in any variable $k$, it follows that every restriction of $f$ must also be canalizing in variable $k$, i.e. $\hat{f}^{(U)}(\emptyset) +  a_k\cdot \hat{f}^{(U)}(\{k\}) = b$, we get:
\begin{align}\nonumber
b &= \sum_{S \subseteq U} \left( \chi_{S}(\bar{\alpha} ) \cdot  \hat{f}(S) \right) + a_k\cdot \sum_{S \subseteq U} \left( \chi_{S}(\bar{\alpha} )\cdot  \hat{f}(\{k\}\cup S) \right), \\\nonumber
b  &= \sum_{S \subseteq U, S \neq \emptyset} \left(-a_j\cdot \chi_{S\setminus\{j\}}(\bar{\alpha} ) \cdot  \hat{f}(S) \right) + \hat{f}(\emptyset) \\ \nonumber 
&+ a_k\cdot \sum_{S \subseteq U, S \neq \emptyset} \left(\chi_{S}(\bar{\alpha} ) \cdot  \hat{f}(\{k\}\cup S) \right)+ a_k\cdot \hat{f}(\{k\}).
\end{align}
Using Eq. \reff{eq:cond:1} and the induction hypothesis, we get
\begin{align}\nonumber
0 &=\sum_{S \subseteq U, S \neq \emptyset} \left( - c \right)  +  \sum_{S \subset U, S \neq \emptyset} \left( a_k \cdot \chi_{S}(\bar{\alpha} ) \cdot  \hat{f}(\{k\}\cup S) \right)  \\ \nonumber
&+\left( a_k \cdot \chi_{U}(\bar{\alpha} ) \cdot  \hat{f}(\{k\}\cup U) \right).
\end{align}
We again assume \reff{eq:prop:most:2} holds for all $S \subset U$, i.e. $|S| < |U|$ and, hence:
\begin{align}\nonumber
0 &= - (2^{|U|}-1)\cdot c  + (2^{|U|}-2)\cdot c + \left( a_k \cdot \chi_{U}(\bar{\alpha} ) \cdot  \hat{f}(\{k\}\cup U) \right)\\\nonumber
c &= a_k \cdot \chi_{U}(\bar{\alpha} ) \cdot  \hat{f}(\{k\}\cup U),
\end{align}
which concludes the proof.
\end{proof}

For the special case, in which all variables are most dominant, we derive the following corollaries:
\begin{corollary}\label{cor:most:a}
Let $f$ be a $\{\pi:\alpha:\beta\}$ NCF with $n$ variables of which $k$ are relevant. All variables are most canalizing if the Fourier coefficients fulfill the following condition,
\begin{align} \label{eq:corr:most:a:1}
\alpha_j \cdot \left(\prod_{i \in S,i\neq j}{\bar{\alpha}_i} \right)\hat{f}(S) = c \hspace{1cm} \forall S \subseteq [n], S \neq \{\emptyset\} \text{ and } \forall j \in S
\end{align}
with
\begin{align} \label{eq:corr:most:a:c}
c= 2^{-(k-1)}.
\end{align}
\end{corollary}
\begin{proof}
Eq.~\reff{eq:corr:most:a:1} follow directly from Proposition \ref{prop:most}, while Eq. \reff{eq:corr:most:a:c} follows from Parsevals theorem.
\end{proof}
Corollary~\ref{cor:most:a} can easily be rewritten as:

\begin{corollary}\label{cor:most:b}
Let $f$ be a $\{\pi:\alpha:\beta\}$ NCF with $n$ variables of which $k$ are relevant. All variables are most canalizing if the absolute values of the Fourier coefficients fulfill the following condition,
\begin{align} \label{eq:corr:most:b:2}
\left|\hat{f}(S) \right| &= c \hspace{1cm} \forall S \subseteq K, S \neq \{\emptyset\},\\ \nonumber
\left|\hat{f}(\emptyset) \right| &= 1- c
\end{align}
with
\begin{align} 
c= 2^{-(k-1)}.
\end{align}
\end{corollary}

\begin{corollary}
Let $f$ be a $\{\pi:\alpha:\beta\}$ NCF with $k>1$ relevant input variables. All variables are most canalizing and $\beta_i=b, \forall i\in \{1,
\ldots,k\}$. All such NCFs are completely described by $\alpha$ and $b$.
and hence there are $2^{(k+1)}$ such functions.
\end{corollary}
\begin{proof}
The statement follows directly from the previous corollary.
\end{proof}
As mentioned before, the zero coefficient plays an important role. Interestingly, we can describe the zero coefficients for NCFs in a recursive manner:
\begin{proposition} \label{prop:zero}
The zero coefficient of a $\{\pi:\alpha:\beta\}$ NFC $f$ can be recursively written as :
$$
    \hat f(\emptyset) =  \frac{1}{2} {\hat f}^{(\pi_1,\alpha_1)}(\emptyset) + \frac{1}{2}\beta_{1}.
$$
\end{proposition}
\begin{proof}
Follows directly from Corollary \ref{cor:zero}.
\end{proof}
Further, the zero coefficient is upper bounded as shown by following  proposition:
\begin{proposition} \label{prop:zero_up_ncf}
The absolute value of the zero coefficient of a $\{\pi:\alpha:\beta\}$ NCF $f$ with $k>1$ relevant and uniformly distributed input variables can be bounded as:
$$
\frac{1}{2^{k-1}} \leq |\hat f(\emptyset)| \leq 1-\frac{1}{2^{k-1}}.
$$
\end{proposition}
\begin{proof} First, we prove the right hand side:
Using the triangle inequality we get from Proposition \ref{prop:zero}:
$$  
|\hat f(\emptyset)| \leq  \frac{1}{2} |\hat f^{(\pi_{i},\alpha_{i})}(\emptyset)| + \frac{1}{2} .
$$
Obviously the zero coefficient of a function with only one relevant variable $i$ is zero. The proposition now follows by induction.
The left hand side can be easily shown using the inverse triangle inequality and induction.
\end{proof}

As seen in Corollary~\ref{cor:most:b}, a NCF, whose variables are most dominant, fulfills the upper bound with equality. The following proposition follows directly from Proposition~\ref{prop:zero}:
\begin{proposition}
Let $f$ be a NCFs with $k$ relevant variables and alternating $\beta_i$, i.e., with $\beta=(-1, +1, -1, +1, \ldots)$ or $\beta=(+1,-1,+1,-1,\ldots)$, then the absolute value of the zero coefficient is given as:
$$
|\hat f(\emptyset)| =  \frac{1}{3} \left(\frac{1}{2^{k-1}} (-1)^k + 1\right)
$$
\end{proposition}

\subsection*{Average Sensitivity}
\subsubsection*{Definition}
The average sensitivity (as) \cite{BKS99} is a measure to quantify the influence of random perturbations of the inputs of Boolean functions. It is defined as the sum of the influences of the inputs of the function, which is defined as the probability that the function's output changes, if input $i$ is flipped:
\begin{definition} (\cite{HSB,BL85})
The influence of variable $i$ on the function $f$ is defined as
$$
I_i(f)=\Pr{f(\X)\neq f(\X\oplus e_i)}.
$$
\end{definition}
The influence can be related to the Fourier spectra as follows \cite{BT96}:
$$
I_i(f) = \sum_{S\subseteq[n]:i\in S}\hat f(S)^2.
$$

The average sensitivity is defined as the sum of the influences of all input variables of $f$.
\begin{definition} (\cite{HSB,BKS99})
The average sensitivity of $f$ is defined as
$$
\as(f)=\sum_{i\in[n]} I_i(f).
$$
\end{definition}
Consequently the average sensitivity can also be expressed in terms of the Fourier coefficients \cite{HSB} as:
\begin{equation} \label{eq:as_fc}
\as(f)=\sum_{S\subseteq[n], S\neq \emptyset}\hat f(S)^2 |S|.
\end{equation}

\subsubsection*{Restricted Functions}
To investigate the average sensitivity of restricted functions we first define $\xi:\mathcal F_n \times \mathcal F_n \rightarrow \mathbb R$ by
\begin{equation}
    \xi(f,g)  = \frac{1}{2}\left(1 - \sum_{U \subseteq [n]} \hat f(U) \hat g(U)\right).
    \label{eq:def_xi}
\end{equation}
Our next result shows the relation between the average sensitivity of a BF and the average sensitivity of its two restricted functions.
\begin{theorem}\label{theorem:1}
Let $f^{(i,+)}, f^{(i,-)}$ be the restrictions of $f$ to some relevant variable $i$ of $f$. Then 
\begin{align*}\nonumber      
  \as(f)  =&\frac{1}{2} \as(f^{(i,+)})  + \frac{1}{2}\as(f^{(i,-)})+ \xi(f^{(i,+)},f^{(i,-)}).
\end{align*}
\end{theorem}
\begin{proof}
Starting from Eq. \reff{eq:as_fc}, we can fractionize the Fourier coefficients according to Proposition \ref{prop:res2}. This yields:
\begin{align*}
    \as (f) =& \sum_{S\subseteq[n], S\neq \emptyset} \left(\frac{1}{2}
    \hat f^{(i,+)}( S \setminus \{i\}) + \frac{1}{2}\left(-1\right)^{| S \cap \{i\}|  } \hat f^{(i,-)}( S \setminus \{i\})  \right)^2 |S| \\
    =& \frac{1}{4}\sum_{S\subseteq[n], S\neq \emptyset} \bigg(\left(
    \hat f^{(i,+)}( S \setminus \{i\})   \right)^2
    +  \left(\left(-1\right)^{| S \cap \{i\}|  } \hat f^{(i,-)}( S \setminus \{i\})  \right)^2 \\
    &+ 2 \left(-1 \right) ^{| S \cap \{i\}|} 
    \hat f^{(i,+)}( S \setminus \{i\})\hat f^{(i,-)}( S \setminus \{i\}) \bigg) |S|
\end{align*}
which leads us to:
\begin{align*}
    \as (f) =& \frac{1}{4}\sum_{S\subseteq[n], S\neq \emptyset} \left( 
    \hat f^{(i,+)}( S \setminus \{i\})   \right)^2 |S| \\
    +&  \frac{1}{4}\sum_{S\subseteq[n], S\neq \emptyset}  \left( \hat f^{(i,-)}( S \setminus \{i\})  \right)^2 |S| \\ \nonumber
     +& \frac{1}{2}\sum_{S\subseteq[n], S\neq \emptyset} \left(-1 \right) ^{| S \cap \{i\}|} 
    \hat f^{(i,+)}( S \setminus \{i\})\hat f^{(i,-)}( S \setminus \{i\}) |S|,
\end{align*}
and hence to:
\begin{align*}\nonumber
    \as (f) =& \frac{1}{4} \sum_{S\subseteq[n]\setminus \{i\}, S\neq \emptyset} \left(
    \hat f^{(i,+)}( S)   \right)^2 |S| + \frac{1}{4} \sum_{S\subseteq[n]\setminus\{i\}} \left( 
    \hat f^{(i,+)}( S )   \right)^2 \left(1+ |S|\right) \\ \nonumber
    +& \frac{1}{4} \sum_{S\subseteq[n]\setminus\{i\}, S\neq \emptyset}  \left(\hat f^{(i,-)}( S)  \right)^2 |S| + \frac{1}{4} \sum_{S\subseteq[n]\setminus\{i\}}  \left(\hat f^{(i,-)}( S)  \right)^2 \left(1+|S|\right) \\ \nonumber
+& \frac{1}{2} \sum_{S\subseteq[n]\setminus\{i\}, S\neq \emptyset}\left(-1 \right) ^{0} 
    \hat f^{(i,+)}( S)\hat f^{(i,-)}( S) |S| \\ \nonumber
    +& \frac{1}{2} \sum_{S\subseteq[n]\setminus\{i\}} \left(-1 \right) ^{1} 
    \hat f^{(i,+)}( S)\hat f^{(i,-)}( S) \left(1 + |S|\right).
\end{align*}
Since $f^{(i,a)}(S) = 0$ for all $S: i \in S$ we can write 
\begin{align*}\nonumber
    \as (f) =&\frac{1}{2} \underbrace{\sum_{S\subseteq[n], S\neq \emptyset} \left(
    \hat f^{(i,+)}( S)\right)^2 |S|}_{=\as(f^{(i,+)})} + \frac{1}{4} \underbrace{\sum_{S\subseteq[n]} \left( \hat f^{(i,+)}( S )   \right)^2}_{=1} \\ \nonumber
    +&\frac{1}{2}  \underbrace{\sum_{S\subseteq[n], S\neq \emptyset} \left( \hat f^{(i,-)}( S)  \right)^2 |S|}_{=\as(f^{(i,-)})} + \frac{1}{4} \underbrace{\sum_{S\subseteq[n]} \left(\hat f^{(i,-)}( S)  \right)^2}_{=1}\\ \nonumber
    +& \frac{1}{2} \sum_{S\subseteq[n], S\neq \emptyset} 
    \hat f^{(i,+)}( S)\hat f^{(i,-)}( S) |S| 
    - \frac{1}{2} \sum_{S\subseteq[n],S\neq \emptyset}
    \hat f^{(i,+)}( S)\hat f^{(i,-)}( S) |S| \\ \nonumber
    -& \frac{1}{2} \sum_{S\subseteq[n]}   \hat f^{(i,+)}( S)\hat f^{(i,-)}(S).
\end{align*}
Finally we get      
\begin{align*}\nonumber      
  \as(f)  =&\frac{1}{2} \as(f^{(i,+)})  +\frac{1}{2}\as(f^{(i,-)})
     + \frac{1}{2} 
     -\frac{1}{2}   \sum_{S\subseteq[n]}  \hat f^{(i,+)}( S)\hat f^{(i,-)}( S),
\end{align*}
which concludes the proof.
\end{proof}

For NCFs we obtain:
\begin{corollary} \label{cor:as_ncf}
The average sensitivity of a $\{\pi:\alpha:\beta\}$ NCF can recursively be described as:
\begin{equation}
    \as(f) =  \frac{1}{2} \left( \as(f^{(\pi_{i},\alpha_{i})})  + 1- \hat f^{(\pi_{i},\alpha_{i})}(\emptyset) \beta_{i} \right).
\end{equation}
\end{corollary}

In \cite{LAM13} an upper bound on the average sensitivity of NCF has been conjectured. In the following theorem, we prove this conjecture to be correct.
\begin{theorem}\label{theorem:ncfbound}
The average sensitivity of a NCF with $k=rel(f)$ relevant and uniform distributed variables is bounded by 
\begin{equation} \label{eq:prop:ncfbound}
    \frac{k}{2^{k-1}} \leq \as(f) \leq \frac{4}{3} -2^{-k}-\frac{1}{3}\cdot 2^{-k}\cdot(-1)^{k}.
\end{equation}
\end{theorem}
The bounds in Eq.~\reff{eq:prop:ncfbound} will turn out to be tight.
\begin{proof}
We first prove the upper bound in Eq.~\reff{eq:prop:ncfbound}. Let us recall Corollary~\ref{cor:as_ncf}:
$$
\as(f) =  \frac{1}{2} \left( \as(f^{(\pi_{i},\alpha_{i})})  + 1- \hat f^{(\pi_{i},\alpha_{i})}(\emptyset) \beta_{i} \right).
 $$
If we apply Corollary~\ref{cor:as_ncf} again and use Proposition \ref{prop:zero}, we get:
\begin{align}\nonumber
\as(f) =&  \frac{1}{2} \left( \frac{1}{2} \left( \as({f^{(\pi_{i},\alpha_{i})}}^{(\pi_{i+1},\alpha_{i+1})})  + 1- {\hat f}^{(\pi_{i},\alpha_{i})^{(\pi_{i+1},\alpha_{i+1})}}(\emptyset) \beta_{i+1} \right) \right.\\ \nonumber
&+ \left. 1- \left(\frac{1}{2}
    {\hat f}^{(\pi_i,\alpha_i)^{(\pi_{i+1},\alpha_{i+1})}}(\emptyset) +   \frac{1}{2}\beta_{i+1} \right) \beta_{i} \right) \\ \nonumber
    =&  \frac{1}{4}\as({f^{(\pi_{i},\alpha_{i})}}^{(\pi_{i+1},\alpha_{i+1})})  +  \frac{3}{4}-  \frac{1}{4}{\hat f}^{(\pi_{i},\alpha_{i})^{(\pi_{i+1},\alpha_{i+1})}}(\emptyset) \beta_{i+1} \\ \nonumber
&- \frac{1}{4}
    {\hat f}^{(\pi_i,\alpha_i)^{(\pi_{i+1},\alpha_{i+1})}}(\emptyset)\beta_{i}-   \frac{1}{4}\beta_{i+1}\beta_{i} \\ \nonumber 
 =&  \frac{1}{4}\as({f^{(\pi_{i},\alpha_{i})}}^{(\pi_{i+1},\alpha_{i+1})})  -  \frac{1}{4}{\hat f}^{(\pi_{i},\alpha_{i})^{(\pi_{i+1},\alpha_{i+1})}}(\emptyset)(\beta_{i}+\beta_{i+1}) \\ \nonumber
&-\frac{1}{4}\beta_{i+1}\beta_{i} +  \frac{3}{4} 
\end{align}
Since $\beta_{i},\beta_{i+1}  \in \{-1,+1\}$ and $|{\hat f}^{(\pi_{i},\alpha_{i})^{(\pi_{i+1},\alpha_{i+1})}}(\emptyset)| \leq 1$, 

$$
 -\frac{1}{4}{\hat f}^{(\pi_{i},\alpha_{i})^{(\pi_{i+1},\alpha_{i+1})}}(\emptyset)(\beta_{i}+\beta_{i+1}) -\frac{1}{4}\beta_{i+1}\beta_{i} \leq \frac{1}{4}.
$$
Thus we obtain
\begin{equation} \label{eq:proof:as1}
    \as(f) \leq  \frac{1}{4} \as({f^{(\pi_{i},\alpha_{i})}}^{(\pi_{i+1},\alpha_{i+1})})  +  1,
\end{equation}
where ${f^{(\pi_{i},\alpha_{i})}}^{(\pi_{i+1},\alpha_{i+1})}$ has $k-2$ relevant variables.
We will now show the theorem by induction.
For $k=1$ the upper bound in Eq.~\reff{eq:prop:ncfbound} simplifies to 
$$
\as(f)\leq 1,
$$
which is obviously true.
For $k=2$ the upper bound in Eq.~\reff{eq:prop:ncfbound} results in 
$$
\as(f)\leq 1,
$$
which is also true and can be verified by inspecting all possible functions.

Using Eq.~\reff{eq:prop:ncfbound} as the induction hypothesis, and applying it on ${f^{(\pi_{i},\alpha_{i})}}^{(\pi_{i+1},\alpha_{i+1})}$ in Eq.~\reff{eq:proof:as1}, which has $k-2$ relevant variables, yields:

\begin{align*}
   \as(f) &\leq \frac{1}{4} \left(\frac{4}{3} -2^{-(k-2)}-\frac{1}{3}\cdot 2^{-(k-2)}\cdot(-1)^{k-2}\right) +  1 \\ \nonumber
   &= \frac{4}{3}-2^{-k}-\frac{1}{3} 2^{-k}(-1)^{k} , 
\end{align*}
which concludes the induction.

The lower bound in Eq.~\reff{eq:prop:ncfbound} is commonly known in literature and can be proven along the lines of the proof of the upper bound, using the following inequality, which follows from Corollary~\ref{cor:as_ncf} and Proposition~\ref{prop:zero_up_ncf}:
$$
\as(f) \geq  \frac{1}{2} \left( \as(f^{(\pi_{i},\alpha_{i})})  + \frac{1}{2^{k-2}} \right).
$$
\end{proof}
The tightness of Eq.~\reff{eq:prop:ncfbound} is shown in Propositions~\ref{prop:tightlow} and~\ref{prop:tightup}.

%We can generalize the upper bound of Theorem \ref{theorem:ncfbound} as:
We can further upper bound the right hand side of Theorem~\ref{theorem:ncfbound} in order to make it independent of the number of relevant variables~$k$:
\begin{corollary}\label{cor:43}
The average sensitivity of a NCF with uniformly distributed variables satisfies
\begin{equation*}
 \as(f) < \frac{4}{3}.
\end{equation*}
\end{corollary}

We next show that the bounds in Theorem~\ref{theorem:ncfbound} are tight.
\begin{proposition} \label{prop:tightlow}
Let $f$ be a NCF, whose variables are all most dominant. Then $f$ satisfies the upper bound in Theorem \ref{theorem:ncfbound} with equality.
\end{proposition}
\begin{proof}
Starting from Corollary~\ref{cor:as_ncf}, and using that, by Corollary~\ref{cor:most:b}, $|\hat f(\emptyset)| = 1 - \frac{1}{2^{k-1}}$ and $\beta_{i}=sgn(\hat f(\emptyset))$ for all $i$, we get:
\begin{align}\label{proof:cor:as}
\as(f) &=  \frac{1}{2} \left( \as(f^{(\pi_{i},\alpha_{i})})  + 1- (1 - \frac{1}{2^{k-2}})\right) \\
&= \frac{1}{2} \left( \as(f^{(\pi_{i},\alpha_{i})}) + \frac{1}{2^{k-2}})\right).
\end{align}
Since $\as(f)$ depends on $k$ relevant variables, while $\as(f^{(\pi_{i},\alpha_{i})})$ depends only on $k-1$ relevant variables, Eq.~\reff{proof:cor:as} becomes:
\begin{align*}
\as(k) &= \frac{1}{2} \left( \as(k-1) + \frac{1}{2^{k-2}})\right).
\end{align*}
The proof is concluded by solving this recursion using induction.
\end{proof}

\begin{proposition} \label{prop:tightup}
Let $f$ be a NCF with alternating $\beta_i$,  i.e. $\beta=(-1,+1,-1,+1,\ldots)$ or $\beta=(+1,-1,+1,-1,\ldots)$. Then $f$ fulfills the upper bound in Eq.~\reff{eq:prop:ncfbound} of Theorem~\ref{theorem:ncfbound} with equality.
\end{proposition}
\begin{proof}
Similar as for the proof of the previous proposition we start from Corollary~\ref{cor:as_ncf} and use $|\hat f(\emptyset)| = \frac{1}{3} \left(\frac{1}{2^{k-1}} (-1)^k + 1\right)$. The proof is established by solving the recursion.
\end{proof}

Propositions~\ref{prop:tightlow} and \ref{prop:tightup} show that the maximal and minimal average sensitivity is achieved if the absolute value of the zero coefficient is minimal and maximal, respectively. The following proposition gives an bound on the average sensitivity for a fixed $|\hat f(\emptyset)|$.
\begin{proposition} \label{prop:combined}
Let $f$ be a NCF. Then 
$$
\as(f) \leq \frac{5}{3} - |\hat f(\emptyset)|.
$$
\end{proposition}
\begin{proof}
Combining Corollaries \ref{cor:as_ncf} and \ref{cor:43}, we get:
$$
\as_(f^{(\pi_{i},\alpha_{i})})  - \beta_i \cdot \hat f^{(\pi_{i},\alpha_{i})}(\emptyset) \leq \frac{5}{3},
$$
and since $\beta_i \in \{-1,+1\}$:
$$
\as(f^{(\pi_{i},\alpha_{i})}) + |\hat f^{(\pi_{i},\alpha_{i})}(\emptyset)|\leq \frac{5}{3}.
$$
Substituting $f^{(\pi_{i},\alpha_{i})}$ by $f$ concludes the proof.
\end{proof}

\begin{figure}
\centering
\begin{tikzpicture}[domain=0:4]
%    \draw[->] (-4.2,0) -- (4.2,0) node[right] {$\hat f(\emptyset)$};
%    \draw[->] (0, 0) -- (0,4.2) node[above] {$\as(f)$};
\begin{axis}[
ylabel={$\as(f)$},
xlabel={$\hat f(\emptyset)$},
legend style={at={(+0.5,1.05)}, anchor= south, legend columns=2}
]

\pattern[pattern=north west lines,pattern color=black] (axis cs:-1,0.687500) -- (axis cs:-0.333,1.333) --(axis cs:0.333,1.333) -- (axis cs:1,0.687500) -- (axis cs:0.577,0.687500) -- (axis cs:0.480000,0.769600) -- (axis cs:0.380000,0.855600) -- (axis cs:0.280000,0.921600 ) -- (axis cs:0.180000,0.967600 ) -- (axis cs:0,1) -- (axis cs:-0.180000,0.967600 ) --(axis cs:-0.380000,0.855600) -- (axis cs:-0.280000,0.921600 )-- (axis cs:-0.480000,0.769600) -- (axis cs:-0.577,0.687500);

\pattern[pattern=dots,pattern color=black] (axis cs:-1,0.0) --(axis cs:-1,0.687500) -- (axis cs:-0.3 33,1.333) --(axis cs:0.333,1.333) -- (axis cs:1,0.687500) -- (axis cs:1,0.0) -- (axis cs:0.780000,0.391600) -- (axis cs:0.577,0.687500) -- (axis cs:0.480000,0.769600) -- (axis cs:0.380000,0.855600) -- (axis cs:0.280000,0.921600 ) -- (axis cs:0.180000,0.967600 ) -- (axis cs:0,1) -- (axis cs:-0.180000,0.967600 ) --(axis cs:-0.380000,0.855600) -- (axis cs:-0.280000,0.921600 )-- (axis cs:-0.480000,0.769600) -- (axis cs:-0.577,0.687500) -- (axis cs:-0.780000,0.391600);

\addplot[smooth,solid] file {./theo1_up_bd.dat};
%\addlegendentry{Corollary \ref{cor:43}}
\addlegendentry{Corollary 5}

\addplot[smooth,densely dashed] file {./theo1_low_bd.dat};
%\addlegendentry{Theorem \ref{theorem:ncfbound}, k=5 (LHS)}
\addlegendentry{Theorem 2, k=5 (LHS)}
\addplot[smooth,dotted] file {./gen_low_bd.dat};
%\addlegendentry{lower bound \cite{KFG11}}
\addlegendentry{lower bound [15]}
\addplot[smooth,loosely dashed] file {./prop9_up_bd.dat};
%\addlegendentry{Proposition \ref{prop:combined}}
\addlegendentry{Proposition 8}

\end{axis}

\end{tikzpicture}
\end{figure}

\section*{Discussion}
In Figure 1 we summarize the bounds from the previous section. We plot the average sensitivity versus the zero coefficient. Additionally we include a lower bound on the average sensitivity that is independent of the number of relevant variables, and can be found in~\cite{KKL88}. One can see that this bound intersects with our lower bound (which we plotted for $k=5$), though we stated that our bound is tight. However, this is not a contradiction, since the lower bound in Theorem \ref{theorem:ncfbound} is achieved for functions with a large absolute zero coefficients, which are located outside the intersection.

For $k=5$ our lower bound forms a triangle with the upper bound as formulated in Proposition \ref{prop:combined}. The NCFs with all variables being most dominant are located in the left and right corners of that triangle. However, for larger $k$ the lower bound will further decrease, and with it the most dominant NCFs. 

The upper bound in Corollary~\ref{cor:43} also intersects with the bound of Proposition \ref{prop:combined}. Again, this is not a contradiction, since NCFs reach this bound only for small absolute zero coefficients. 

In general the average sensitivity is upper bounded by $k$, i.e., $\as(f) \leq k$. As shown in \cite{FK96} for monotone and in \cite{KFS13} for unate, i.e. locally monotone, functions, the average sensitivity is upper bounded by $\as(f) \leq \sqrt{(1- \hat f(\emptyset))k}$. This bound is tight up to a multiplicative constant, see e.g. \cite{MO03}. A function is unate, if it is monotone in each variable. In a regulatory network, where each regulator acts either inhibitory or exhibitory towards a certain gene, each function will be unate. NCFs form a subclass of the unate functions. Thus, our results show, that even within the class of unate functions, the average sensitivity of NCFs is remarkably low. Since a low average sensitivity has a positive effect on the stability of Boolean networks \cite{SK04}, our result gives an explanation for the remarkable stability of BNs with NCFs.

\section*{Conclusion}
In this paper we investigated Boolean functions, in particular canalizing and nested canalizing functions,  using Fourier analysis.
We gave recursive representations for the zero coefficient and the average sensitivity based on the concept of restricted BFs. 

We addressed the average sensitivity of nested canalizing functions and provided a tight upper and lower bound on the average sensitivity. We showed that the lower bound is achieved by functions whose input variables are all most dominant and which maximize the absolute zero coefficient. The upper bound is reached by functions, whose canalized values are alternating.

We provided an upper bound on the average sensitivity, namely $\as(f)< \frac{4}{3}$, which has been conjectured in literature\cite{LAM13}. Finally we derived a bound on the absolute zero coefficient and  the average sensitivity and discussed the stabilizing effect of nested canalizing functions on the network dynamics.

It is worth noting that all those results rely on the assumption of uniformly distributed inputs. 
This rises the question if the results can be generalized to other distributions. 
The recursive representations can easily be extended to product distributed input variables. 
But without further constraints  there always exists a  distribution which maximizes the 
(accordingly defined) average sensitivity, i.e., for any function with $k$ relevant variables
the average sensitivity can be $k$. 

\section*{Acknowledgments}
The authors want to thank Vladimir Sidorenko for fruitful discussions and proofreading the manuscript.
This work was supported by the German research council "Deutsche Forschungsgemeinschaft" (DFG) under Grant Bo 867/25-2.

 \bibliographystyle{plos2009}  % Style BST file
  \bibliography{literatur}     % Bibliography file (usually '*.bib' ) 

\section*{Figure Legends}

\subsection*{Figure 1}
\textbf{Important Bounds:} The dotted-area corresponds to the possible values for the average sensitivity of a NCF, the lined area to BFs with $k=5$ input variables.

\end{document}